\documentclass[12pt]{article} 
\usepackage[cp1251]{inputenc}
\usepackage{graphicx}
\usepackage{cite}
\usepackage{textcase}
\usepackage{amsfonts}
\usepackage{amsmath}
\usepackage[russian]{babel}

\usepackage{MnSymbol}%
\usepackage{wasysym}%

\usepackage{xcolor}
\usepackage{float}

\newtheorem{theorem}{Теорема}
\newtheorem{lemma}{Лемма}

\newtheorem{proof}{Доказательство}
\newtheorem{corollary}{Следствие}
\newtheorem{definition}{Определение}

\begin{document}  

\begin{center}
\textbf{Обобщение теоремы Гершгорина. Анализ и синтез систем управления}%
\end{center}

\begin{center}
 Игорь Борисович~Фуртат
\end{center}

\begin{center}
Институт проблем машиноведения РАН (ИПМаш РАН), 2024 г.
\end{center}

\begin{abstract} 
Рассматривается применение теоремы кругов Гершгорина и некоторых производных от нее результатов для оценки собственных значений матрицы. 
Развиваются полученные результаты для получения области локализации собственных значений матрицы с интервально неопределенными постоянными или нестационарными элементами. 
Вводится понятие $e$-кругов для получения более точных оценок данных областей, чем при использовании кругов Гершгорина. 
Полученные результаты применены к анализу устойчивости сетевых систем, где показано, что предложенные методы позволяют анализировать сеть с гораздо большим числом агентов, чем при использовании методов CVX, Yalmip, eig и lyap (функции в MatLab). 
Далее показано, что если полученные результаты применять не к самой системе, а к результату, полученному с помощью метода функций Ляпунова, то можно исследовать системы с матрицами без диагонального преобладания. 
Это позволило рассмотреть модификацию условия Демидовича на системы с нестационарными параметрами и синтез закона управления для нестационарных систем с матрицами без диагонального преобладания. 
Все полученные результаты иллюстрируются численным моделированием.

\end{abstract}

\textit{Ключевые слова:} теорема Гершгорина, область локализации собственных значений матрицы, устойчивость, управление.

\section{Введение}

При анализе динамических систем и синтезе закона управления ими одним из ключевых вопросов является устойчива ли система. 
В настоящее время для определения устойчивости применяется много различных методов и подходов: вычисление собственных значений матрицы \cite{Voronov86}, различные алгебраические и частотные критерии устойчивости \cite{Voronov86}, метод функций Ляпунова \cite{Voronov86}, дивергентные методы исследования устойчивости \cite{Furtat21} и т.д.

Данная статья сконцентрирована на исследовании локализации собственных значений матрицы с некоторыми применениями к анализу и синтезу систем управления. 
Для построения области локализации собственных значений будет рассмотрена теорема Гершгорина \cite{Belman69,Khorn89,Gantmakher10} и некоторые следствия из нее, а также будут получены новые результаты по обобщению данной теоремы на случай параметрически неопределенных матриц и матриц с нестационарными параметрами. 

Теорема Гершгорина и различные ее модификации неоднократно рассматривались в литературе. 
Интерес к данному аппарату связан с простым способом определения области локализации собственных значений. 
Зачастую, теорема Гершгорина приводит к изучению систем с матрицей с диагональным преобладанием (также в литературе встречается название \textit{матрица Адамара}). 
Именно такие системы изучались в работе \cite{Polyak02a,Polyak02b,Polyak04} и были названы \textit{сверхустойчивыми} (если все круги Гершгорина целиком находились в левой полуплоскости комплексной плоскости). 
Было показано, что анализ и синтез систем управления приводит к задачам выпуклой оптимизации. 
В работах \cite{Uronen72,Soloviev83,Curran09,Vijay15,Li19} были получены уточняющие области локализации в виде усредненных оценок, использовании $l_1$ векторных норм и т.д., а в работах \cite{Kazakova98,Vijay14} предложен синтез статических линейных законов управления с использованием теоремы Гершгорина. 
В \cite{Pachauri14,Xie22,Adom23} рассмотрено применение теоремы Гершгорина к исследованию устойчивости моделей в химической промышленности, моделей электрических сетей с трехфазными генераторами и биологических моделей эпидемии.

Анализ литературы показал, что при построении области локализации собственных значений матрицы метод Гершгорина имеет преимущества в простоте его применения, выпуклой процедуре поиска оценок области локализации и небольших вычислительных затратах. 
Однако ограничения применения данного метода связаны с завышенными оценками области локализации и требованием к матрице с диагональным преобладанием (или приводимым к ним с использованием диагональной матрицы для преобразования базиса). 
Требование к системам с матрицами с диагональным преобладанием является особенно ограничительным для возможности синтеза закона управления.

В данной статье будет рассмотрено решение следующих задач:
\begin{enumerate}
\item будут рассмотрены оценки и области локализации собственных значений постоянной матрицы;
\item будут получены области локализации собственных значений матрицы с интервальной параметрической неопределенностью;
\item в качестве примеров применения полученных результатов будут рассмотрены: 
\begin{enumerate}
\item задача синхронизации сетевых систем с большим числом скалярных агентов, где будет показано, что предложенные результаты могут быть применимы для анализа устойчивости гораздо большего числа агентов, чем при использовании методов CVX, Yalmip, eig и lyap (команды в MatLab);
\item модификации условия Демидовича (об устойчивости линейных систем с нестационарными параметрами \cite{Afanasiev03}(теорема 6.1), \cite{Khalil09}) на системы с интервально неопределенными нестационарными параметрами и с матрицей исходной системы без диагонального преобладания;
\item задача поиска матрицы в линейном законе управления с использованием линейных матричных неравенств для объектов с матрицей без диагонального преобладания.
\end{enumerate}
\end{enumerate}

В статье используются следующие \textit{обозначения:} 
$\mathbb C$ -- множество комплексных чисел;
$\mathbb R^n$ -- $n$-мерное евклидово пространство с векторной нормой $|\cdot|$; 
$\mathbb R^{n \times n}$ -- множество всех действительных матриц размерности $n \times m$ с индуцированной матричной нормой $\| Q \|$ для матриц $Q=(q_{ij}) \in \mathbb R^{n \times n}$; 
$\mathbb R_{\geq}$ -- множество неотрицательных вещественных чисел; 
$\lambda_i\{Q\}$ -- $i$-е собственное число квадратной матрицы $Q$;
$\mathbb \Re \{\lambda_i\{Q\} \}$ -- действительная часть $i$-го собственного числа квадратной матрицы $Q$; 
$\mathbb \Im \{\lambda_i\{Q\} \}$ -- мнимая часть $i$-го собственного числа квадратной матрицы $Q$; 
$I$ -- единичная матрица соответствующего порядка.


\section{Оценки областей локализации собственных чисел}
\label{Sec2}

\subsection{Постоянные матрицы}
\label{subsec21}


В данном разделе будут получены оценки области локализации собственных чисел матрицы $Q=(q_{ij}) \in \mathbb R^{n \times n}$ с постоянными элементами. 
Для уточнения данных оценок будет дополнительно рассмотрена диагональная матрица $D=diag\{d_1,...,d_n\}$. 
Введем обозначения сумм по строкам и столбцам абсолютных значений элементов матриц $Q$ и $D^{-1}QD$ без диагональных элементов в виде
\begin{equation}
\label{eq_2_estim_1}
\begin{array}{lll}
&R_i(Q)=\sum\limits_{j=1, j \neq i}^{n}|q_{ij}|,
&C_j(Q)=\sum\limits_{i=1, i \neq j}^{n}|q_{ij}|,
\\
&R_i^D(Q)=\sum\limits_{j=1, j \neq i}^{n} \frac{d_j}{d_i} |q_{ij}|, 
&C_j^D(Q)=\sum\limits_{i=1, i \neq j}^{n} \frac{d_i}{d_j} |a_{ij}|.
\end{array}
\end{equation}

Ниже представлены две леммы, которые позволяют получить оценки снизу и сверху на действительные части собственных чисел матрицы $Q$.

\begin{lemma}
\label{lemma1}
Рассмотрим матрицу $Q \in \mathbb R^{n \times n}$. 
Существуют $d_i>0$, $i=1,...,n$ такие, что справедливы следующие оценки
\begin{equation}
\label{eq_2_estim_1}
\begin{array}{lll}
\max\limits_{i} \{\Re \{\lambda_i\{Q\} \} \} \leq \sigma^D_{\max} \{Q\} \leq \sigma_{\max} \{Q\},
\\
\min\limits_{i} \{\Re \{\lambda_i \{Q\}\} \} \geq \sigma^D_{\min} \{Q\} \geq  \sigma_{\min} \{Q\},
\end{array}
\end{equation}
где
\begin{equation}
\label{eq_2_estim_2}
\begin{array}{lll}
\sigma_{\max}(Q) = \min \left\{ \max\limits_{i}\{q_{ii} + R_i(Q) \}, \max\limits_{j}\{q_{jj} + C_j(Q) \} \right\},
\\
\sigma_{\min}(Q) = \max\{ \min\limits_{i}\{q_{ii} - R_i(Q) \}, \min\limits_{j}\{q_{jj} - C_j(Q) \} \},
\\
\sigma^D_{\max}(Q) = \min\{ \max\limits_{i}\{q_{ii} + R_i^D(Q) \}, 
\max\limits_{j}\{q_{jj} + C_j^D(Q) \} \},
\\
\sigma^D_{\min}(Q) =\max \{ \min\limits_{i}\{q_{ii} - R_i^D(Q) \}, 
\min\limits_{j}\{q_{jj} - C_j^D(Q) \} \}.
\end{array}
\end{equation}
\end{lemma}

 
\begin{proof}
Согласно теореме Гершгорина \cite{Khorn89} все собственные значения матрицы $Q$ заключены в объединении $n$ кругов  
$\cup_{i=1}^{n} \{z \in \mathbb C: |z-q_{ii}| \leq \sum\limits_{j=1, j \neq i}^{n} |q_{ij}| \}$. 
Поскольку матрица $Q^{\rm T}$ имеет те же собственные значения, что и матрица $Q$, то все собственные значения $Q$ также заключены в объединении $n$ кругов  
$\cup_{i=1}^{n} \{z \in \mathbb C: |z-q_{ii}| \leq \sum\limits_{j=1, j \neq i}^{n} |q_{ji}| \}$.
Следовательно, $\sigma_{\min}\{Q\} \leq \min\limits_{i} \{\Re \{\lambda_i \{Q\}\} \}$ и $\sigma_{\max} \{Q\} \geq \max\limits_{i} \{\Re \{\lambda_i\{Q\} \} \}$.

Теперь рассмотрим диагональную матрицу $D=diag\{d_1,...,d_n\}$. 
Известно, что собственные значения матриц $D^{-1}QD$ и $Q$ не меняются. 
Однако за счет варьирования коэффициентов $d_i$ могут быть уменьшены радиусы кругов Гершгорина для матрицы $Q$ в виде 
$\cup_{i=1}^{n} \{z \in \mathbb C: |z-q_{ii}| \leq \sum\limits_{j=1, j \neq i}^{n} \frac{d_j}{d_i}|q_{ij}| \}$ и для матрицы $Q^{\rm T}$ в виде $\cup_{i=1}^{n} \{z \in \mathbb C: |z-q_{ii}| \leq \sum\limits_{j=1, j \neq i}^{n} \frac{d_i}{d_j}|q_{ji}| \}$. 
Следовательно, $\sigma_{\min}^{D}\{Q\} \geq \sigma_{\min}\{Q\}$ и $\sigma_{\max}^D \{Q\} \leq \sigma_{\max} \{Q\}$.
$\blacksquare$
\end{proof}



\begin{lemma}
\label{lemma2}
Рассмотрим матрицу $Q \in \mathbb R^{n \times n}$. 
Существуют $d_i>0$, $i=1,...,n$ и $\alpha \in [0,1]$ такие, что справедливы следующие оценки
\begin{equation}
\label{eq_2_estim_1a}
\begin{array}{lll}
\max\limits_{i} \{\Re \{\lambda_i\{Q\} \} \} \leq \sigma^{D,\alpha}_{\max} \{Q\} \leq \sigma_{\max}^{\alpha} \{Q\} ,
\\
\min\limits_{i} \{\Re \{\lambda_i\{Q\} \} \} \geq \sigma^{D,\alpha}_{\min} \{Q\} \geq  \sigma_{\min}^{\alpha} \{Q\},
\end{array}
\end{equation}
где
\begin{equation}
\label{eq_2_estim_2a}
\begin{array}{lll}
\sigma_{\max}^{\alpha}(Q) = \max\limits_{i,\alpha}\{q_{ii} + [R_i(Q)]^{\alpha} [C_i(Q)]^{1-\alpha} \},
\\
\sigma_{\min}^{\alpha}(Q) = \min\limits_{i,\alpha}\{q_{ii} - [R_i(Q)]^{\alpha} [C_i(Q)]^{1-\alpha} \},
\\
\sigma^{D,\alpha}_{\max}(Q) = \max\limits_{i}\{q_{ii} + [R_i^D(Q)]^{\alpha} [C_i^D(Q)]^{1-\alpha} \},
\\
\sigma^{D,\alpha}_{\min}(Q) = \min\limits_{i}\{q_{ii} - [R_i^D(Q)]^{\alpha}  [C_i^D(Q)]^{1-\alpha}\}.
\end{array}
\end{equation}
\end{lemma}

 
\begin{proof}
Согласно теореме Островского \cite{Khorn89} все собственные числа матрицы $Q$ заключены в объединении $n$ кругов  
$\cup_{i=1}^{n} \{z \in \mathbb C: |z-q_{ii}| \leq [\sum\limits_{j=1, j \neq i}^{n} |q_{ij}|]^{\alpha} [\sum\limits_{j=1, j \neq i}^{n} |q_{ji}|]^{1-\alpha} \}$. 
Следовательно, $\sigma_{\min}^{\alpha}\{Q\} \leq \min\limits_{i} \{\Re \{\lambda_i\{Q\} \} \}$ и $\sigma_{\max}^{\alpha} \{Q\} \geq \max\limits_{i} \{\Re \{\lambda_i\{Q\} \} \}$.

За счет варьирования коэффициентов $d_i$ могут быть уменьшены радиусы кругов 
$\cup_{i=1}^{n} \{z \in \mathbb C: |z-q_{ii}| \leq [\sum\limits_{j=1, j \neq i}^{n} \frac{d_j}{d_i}|q_{ij}|]^{\alpha} [\sum\limits_{j=1, j \neq i}^{n} \frac{d_j}{d_i}|q_{ji}|]^{1-\alpha} \}$. 
Следовательно, $\sigma_{\min}^{D,\alpha}\{Q\} \geq \sigma_{\min}^{\alpha}\{Q\}$ и $\sigma_{\max}^{D,\alpha} \{Q\} \leq \sigma_{\max}^{\alpha} \{Q\}$.
$\blacksquare$
\end{proof}


\begin{corollary}
\label{cor1}
Из формулировок лемм \ref{lemma1} и \ref{lemma2} следует, что для нахождения оценки сверху собственных значений матрицы $Q$ могут использоваться одни значения $d_i$ и $\alpha$, а при нахождении оценки снизу -- другие значения $d_i$ и $\alpha$. 
В результате чего формируется новая уточненная область локализации собственных значений в виде пересечение областей при различных $d_i$ и $\alpha$, поскольку собственные значения должны одновременно принадлежать областям при различных $d_i$ и $\alpha$.
\end{corollary}

\begin{corollary}
\label{cor2}
Из доказательств лемм \ref{lemma1} и \ref{lemma2} также следует, что пересечением соответствующих кругов можно найти область локализации собственных значений матрицы $Q$, откуда можно найти не только оценки сверху и снизу на действительные части собственных значений, но и оценку сверху на мнимую часть, которую обозначим $\hat{\Im} \{Q\} \geq \max\limits_{i} \{ \Im \{\lambda_i\{Q\}\}\}$. 
Значение $\hat{\Im} \{Q\}$ определяется как максимальное значение пересечения кругов вдоль мнимой оси. 
Если оценка сверху действительной части собственного числа матрицы $Q$ имеет отрицательное значение, то можно получить оценку степени колебательности $\mu$ в виде $\mu \leq \hat{\mu}:=\frac{\hat{\Im} \{Q\}}{| \max\limits_{i} \{ \Re \{\lambda_i\{Q\}\}\} |}$. 
Хорошо известно \cite{Voronov86}, что степень колебательности используется для оценки степени перерегулирования в виде $e^{\pi/\mu}$. 
Тогда новая оценка степени перерегулирования определяется как $e^{\pi/\hat{\mu}}$.
\end{corollary}

Сказанное в следствиях \ref{cor1} и \ref{cor2} будет справедливо и в дальнейших обобщениях полученных результатов на возмущенные матрицы.
Продемонстрируем сказанное в леммах и следствиях на следующем примере.


\textit{Пример 1.} 
Рассмотрим матрицу $Q=\begin{bmatrix}
-1 & -2.5\\
-0.5 & -2
\end{bmatrix}$, 
собственные значения которой равны $-1.5 \pm i$. 
В таблице ниже приведены оценки действительной части данного собственного числа. 

\begin{table}[ht]
		\centering
		\begin{tabular}{|c|c|c|c|c|c|c|c|}
			\hline
			Оценка $\mathbb \Re \{\lambda\{Q\}\}$ & Точность, $\%$ \\  
			\hline 
			 $\sigma_{\min}(Q)=-3.5$; $\sigma_{\max}(Q)=0.5$        &    133,3    \\ 
			\hline 
			$\sigma^D_{\min}(Q)=-2.72$; $\sigma^D_{\max}(Q)=-0.27$        &    82    \\ 
			\hline 
			$\sigma^{\alpha}_{\min}(Q)=-2.72$; $\sigma^{\alpha}_{\max}(Q)=-0.27$        &    82    \\ 
			\hline
			$\sigma^{D,\alpha}_{\min}(Q)=-2.26$; $\sigma^{D,\alpha}_{\max}(Q)=-0.73$        &    51,3    \\ 
			\hline 
		\end{tabular} 
        \label{table:nonlin1}
        \caption{Оценки действительно части собственных значений, полученные при использовании \eqref{eq_2_estim_2} и \eqref{eq_2_estim_2a}.}
\end{table}	

\begin{figure}[H]
\begin{minipage}[H]{0.44\linewidth}
\center{\includegraphics[width=1\linewidth]{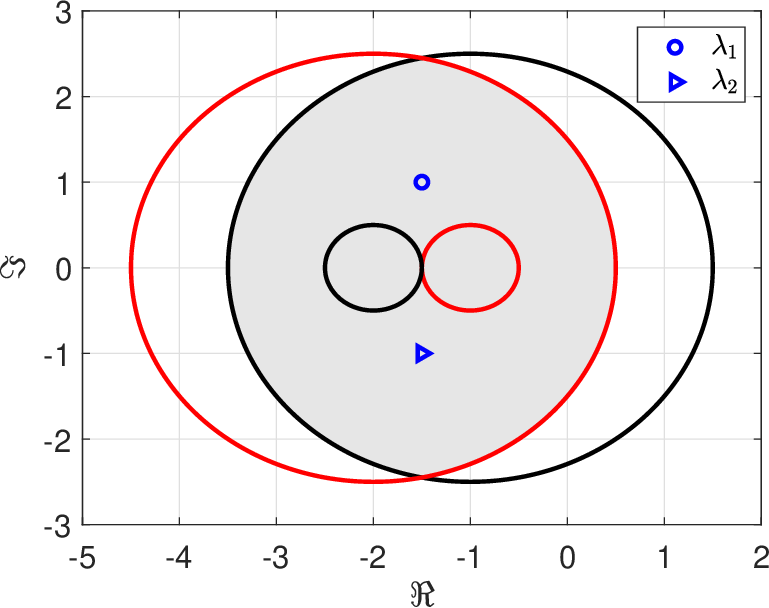}}
\end{minipage}
\hfill
\begin{minipage}[H]{0.44\linewidth}
\center{\includegraphics[width=1\linewidth]{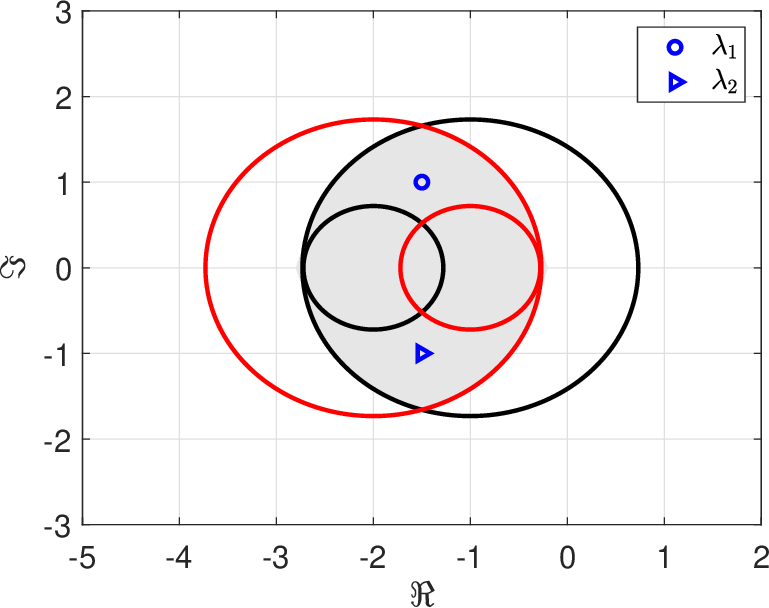}}
\end{minipage}
\vfill
\begin{minipage}[H]{0.44\linewidth}
\center{\includegraphics[width=1\linewidth]{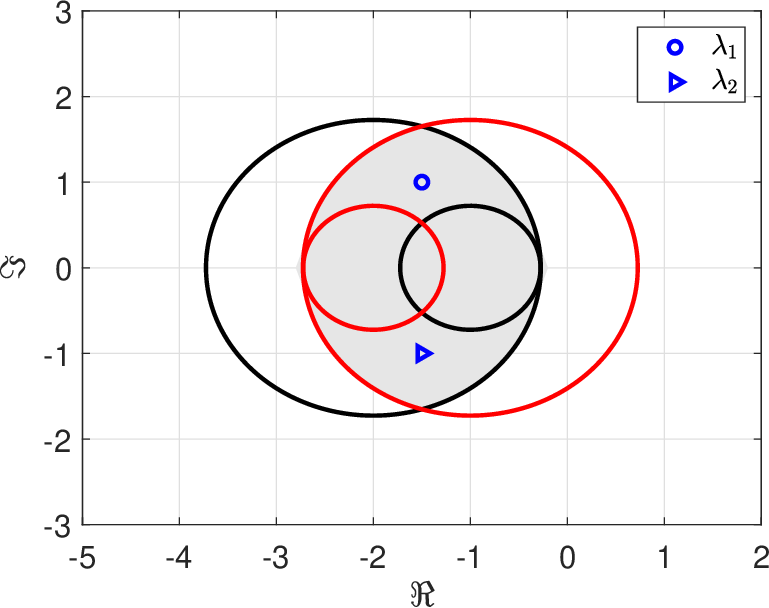}}
\end{minipage}
\hfill
\begin{minipage}[H]{0.44\linewidth}
\center{\includegraphics[width=1\linewidth]{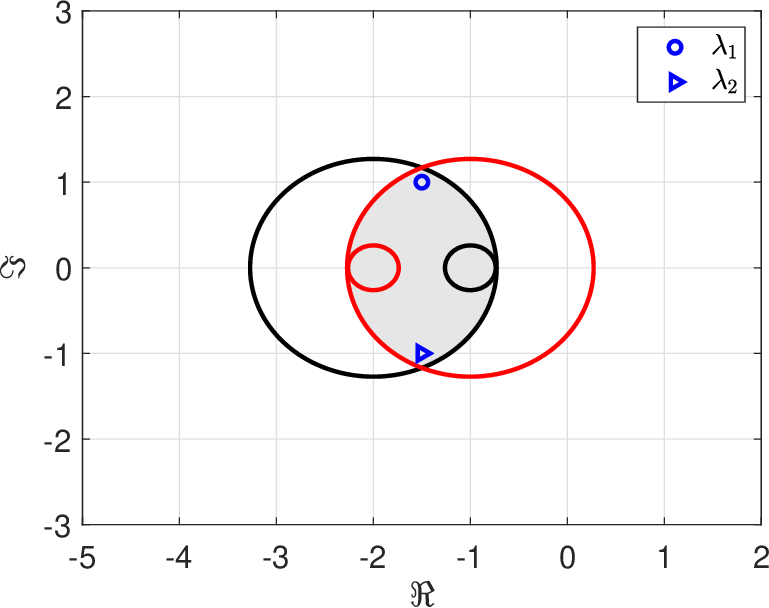}}
\end{minipage}
\caption{Области локализации собственных значений матрицы $Q$ при использовании оценок \eqref{eq_2_estim_2} и \eqref{eq_2_estim_2a}. Результирующая область локализации обозначена серым цветом.}
\label{Ex1}
\end{figure}

Из рис.~\ref{Ex1} можно найти оценки мнимой части, которые отражены в таблице ниже.

\begin{table}[ht]
		\centering
		\begin{tabular}{|c|c|c|c|c|c|c|c|}
			\hline
			Рис.~\ref{Ex1} & Оценка $\mathbb \Im \{\lambda\{Q\}\}$ & Точность, $\%$ \\  
			\hline 
			Сверху слева & 2.3    &    130    \\ 
			\hline 
			Сверху справа & 1.8    &     80    \\ 
			\hline 
			Снизу слева & 1.7    &     70    \\ 
			\hline
			Сверху справа & 1.2    &     20    \\ 
			\hline 
		\end{tabular} 
        \label{table:nonlin2}
        \caption{Оценки мнимой части собственных значений, полученные при использовании лемм \ref{lemma1} и \ref{lemma2}.}
\end{table}	

Наилучшие оценки действительной и мнимой части гарантирует результат из леммы \ref{lemma2}, где одновременно используются варьируемые параметры $D$ и $\alpha$.


%
%


%
%


\subsection{Возмущенные матрицы}
\label{subsec22}

В данном разделе рассмотрим поиск областей локализации собственных чисел для матриц с интервально неопределенными параметрами:
\begin{equation}
\label{eq_2_estim_1d0}
\begin{array}{lll}
&Q(t) = Q_0 + \Delta Q(t) \in \mathbb R^{n \times n},
\\
&Q_0 = (q^0_{ij}), 
&\Delta Q(t) = (\Delta q_{ij}(t)), 
\\
&\Delta \underline{q}_{ii} \leq \Delta q_{ii}(t) \leq \Delta \overline{q}_{ii}, 
&|\Delta q_{ij}(t)| \leq m_{ij}$~\mbox{при}~$i \neq j.  
\end{array}
\end{equation}

Поскольку элементы матрицы могут принимать любые значения из допустимых интервалов, то вместо кругов локализации собственных чисел, рассмотренных в доказательствах лемм \ref{lemma1} и \ref{lemma2}, введем в рассмотрение следующую фигуру.

\begin{definition}
Фигуру, образованную объединением кругов $\mathcal {E C} = \cup_{q \in [\underline{q};\overline{q}]} \{z \in \mathbb C: |z-q| \leq R \}$, назовем $e$-кругом.
\end{definition}

Введем обозначения для оценок сверху сумм по строкам и столбцам абсолютных значений элементов матриц $Q(t)$ и $D^{-1}Q(t)D$, исключая диагональные элементы, в виде
\begin{equation}
\label{eq_2_estim_1d}
\begin{array}{lll}
&\hat{R}_i(Q)=\sum\limits_{j=1, j \neq i}^{n} (|q^0_{ij}|+m_{ij}),
&\hat{C}_j(Q)=\sum\limits_{i=1, i \neq j}^{n} (|q^0_{ij}|+m_{ij}),
\\
&\hat{R}^D_i(Q)=\sum\limits_{j=1, j \neq i}^{n} \frac{d_j}{d_i} (|q^0_{ij}|+m_{ij}),
&\hat{C}^D_j(Q)=\sum\limits_{i=1, i \neq j}^{n} \frac{d_i}{d_j} (|q^0_{ij}|+m_{ij}).
\end{array}
\end{equation}

Теперь рассмотрим обобщение лемм \ref{lemma1} и \ref{lemma2} на случай  матриц с интервально неопределенными элементами. 

\begin{lemma}
\label{lemma3}
Собственные значения матрицы $Q(t)$ из \eqref{eq_2_estim_1d0} находятся в области пересечения $e$-кругов 

\begin{equation}
\label{eq_2_estim_1d}
\begin{array}{lll}
\mathcal{E}\mathcal{C}_{\textup{row}} \cap \mathcal{E}\mathcal{C}_{\textup{col}},
\end{array}
\end{equation}
где
\begin{equation}
\label{eq_2_estim_2d}
\begin{array}{lll}
\mathcal{EC}_{\textup{row}}=\cup_{i=1}^{n} \mathcal{EC}_{\textup{row},i},
\\
\mathcal{EC}_{\textup{row},i}=\cup_{\Delta q_{ii}(t) \in [\underline{q}_{ii};\overline{q}_{ii}]} \{ t \in \mathbb R_{\geq}, \lambda \in \mathbb C: |\lambda(t)-q^0_{ii}-\Delta q_{ii}(t)|
\leq 
\hat{R}^D_i(Q) \},
\end{array}
\end{equation}
\begin{equation}
\label{eq_2_estim_3d}
\begin{array}{lll}
\mathcal{EC}_{\textup{col}}=\cup_{j=1}^{n} \mathcal{EC}_{\textup{col},j},
\\
\mathcal{EC}_{\textup{col},j}=\cup_{\Delta q_{jj}(t) \in [\underline{q}_{jj};\overline{q}_{jj}]} \{ t \in \mathbb R_{\geq}, \lambda \in \mathbb C: |\lambda(t)-q^0_{jj}-\Delta q_{jj}(t)|
\leq 
\hat{C}^D_j(Q) \}.
\end{array}
\end{equation}

\end{lemma}

\begin{proof}
Пусть $\lambda(t)$ -- собственное значение матрицы $Q(t)$ и $s(t)=col\{s_1(t),...,s_n(t)\}$ -- собственный вектор, соответствующий данному собственному значению. 
Выберем $i$-ю компоненту вектора $s(t)$ так, что $\sup\{s_i(t)\} \geq \max\{\sup\{s_1(t)\},...,\sup\{s_{i-1}(t)\},\sup\{s_{i+1}(t)\},...,\sup\{s_n(t)\}\}$. 
Обозначим $\bar{s}_i=\sup\{s_i(t)\}$.
Из соотношения $\lambda(t)s(t)=Q(t)s(t)$ выпишем выражение для $i$-й координаты в виде
$\lambda(t)s_i(t)=\sum\limits_{j=1}^{n} q_{ij}(t) s(t)$ 
или 
$(\lambda(t)-q_{ii}(t))s_i(t)=\sum\limits_{j=1, j \neq i}^{n} q_{ij}(t) s(t)$. 
Воспользовавшись неравенством треугольника, рассмотрим оценку 
\begin{equation}
\label{eq_est_eign1}
\begin{array}{lll}
|\lambda(t)-q_{ii}(t)| |s_i(t)|
=
|\sum\limits_{j=1, j \neq i}^{n} q_{ij}(t) s_j(t)| 
\leq
\\
\leq
\sum\limits_{j=1, j \neq i}^{n} |q_{ij}(t) s_j(t)| 
\leq
\sum\limits_{j=1, j \neq i}^{n} |q_{ij}(t)| |s_j(t)| 
\leq
\bar{s}_i \sum\limits_{j=1, j \neq i}^{n} |q_{ij}(t)|.
\end{array}
\end{equation}

Перепишем выражение \eqref{eq_est_eign1} как 
$|\lambda(t)-q_{ii}(t)| |s_i(t)|-\bar{s}_i \sum\limits_{j=1, j \neq i}^{n} |q_{ij}(t)| \leq 0$ 
или в виде
\begin{equation}
\label{eq_est_eign2}
\begin{array}{lll}
\bar{s}_i \left(|\lambda(t)-q_{ii}(t)| \frac{|s_i(t)|}{\bar{s}_i}
-
 \sum\limits_{j=1, j \neq i}^{n} |q_{ij}(t)| \right)
\leq 0.
\end{array}
\end{equation}
Так как $\frac{|s_i(t)|}{\bar{s}_i} \leq 1$, то выражение \eqref{eq_est_eign2} будет выполнено, если будет выполнено неравенство
\begin{equation}
\label{eq_est_eign3}
\begin{array}{lll}
|\lambda(t)-q_{ii}(t)|
\leq
 \sum\limits_{j=1, j \neq i}^{n} |q_{ij}(t)|.
\end{array}
\end{equation}
Так как
$\Delta \underline{q}_{ii} \leq \Delta q_{ii}(t) \leq \Delta \overline{q}_{ii}$ и $|\Delta q_{ij}(t)| \leq m_{ij}$ при $i \neq j$, то перепишем неравенство \eqref{eq_est_eign3} в виде $e$-круга $\mathcal{EC}_{\textup{row},i}$ из \eqref{eq_2_estim_2d}.

Соотношение \eqref{eq_2_estim_2d} выполнено для некоторого $i$.
Поскольку неизвестно, какое $i$ соответствует данному $\lambda(t)$, то можно лишь сказать, что $\lambda(t)$ принадлежит объединению $e$-кругов 
$\mathcal{EC}_{\textup{row}}=\cup_{i=1}^{n} \mathcal{EC}_{\textup{row},i}$.
Значит, все собственные числа матрицы $Q(t)$ находятся в объединении $e$-кругов $\mathcal{EC}_{\textup{row}}$.

Поскольку матрица $Q^{\rm T}(t)$ имеет те же собственные значения, что и матрица $Q(t)$, то все собственные значения матрицы $Q(t)$ заключены в объединении $e$-кругов $\mathcal{EC}_{\textup{col}}=\cup_{j=1}^{n}\mathcal{EC}_{\textup{col},j}$, см. \eqref{eq_2_estim_3d}. 
Дальнейшие рассуждения для матрицы $Q^{\rm T}(t)$ аналогичны рассуждениям для матрицы $Q(t)$. 
Так как собственные значения матрицы $Q(t)$ находятся одновременно в $\mathcal{EC}_{\textup{row}}$ и $\mathcal{EC}_{\textup{col}}$, значит, они находятся в области \eqref{eq_2_estim_1d}. 
$\blacksquare$
\end{proof}


\begin{lemma}
\label{lemma4}
Пусть заданы $d_i>0$, $i=1,...,n$. 
Собственные значения матрицы $Q(t)$ из \eqref{eq_2_estim_1d0} находятся в области пересечения $e$-кругов 

\begin{equation}
\label{eq_2_estim_1d1}
\begin{array}{lll}
\mathcal{E}\mathcal{C}^D_{\textup{row}} \cap \mathcal{E}\mathcal{C}^D_{\textup{col}},
\end{array}
\end{equation}
где
\begin{equation}
\label{eq_2_estim_2d1}
\begin{array}{lll}
\mathcal{EC}^D_{\textup{row}}=\cup_{i=1}^{n} \mathcal{EC}^D_{\textup{row},i},
\\
\mathcal{EC}^D_{\textup{row},i}=\cup_{\Delta q_{ii}(t) \in [\underline{q}_{ii};\overline{q}_{ii}]} \{ t \in \mathbb R_{\geq}, \lambda \in \mathbb C: |\lambda(t)-q^0_{ii}-\Delta q_{ii}(t)|
\leq 
\hat{R}^D_i(Q) \},
\end{array}
\end{equation}
\begin{equation}
\label{eq_2_estim_3d1}
\begin{array}{lll}
\mathcal{EC}^D_{\textup{col}}=\cup_{j=1}^{n} \mathcal{EC}^D_{\textup{col},j},
\\
\mathcal{EC}^D_{\textup{col},j}=\cup_{\Delta q_{jj}(t) \in [\underline{q}_{jj};\overline{q}_{jj}]} \{ t \in \mathbb R_{\geq}, \lambda \in \mathbb C: |\lambda(t)-q^0_{jj}-\Delta q_{jj}(t)|
\leq 
\hat{C}^D_j(Q) \}.
\end{array}
\end{equation}

\end{lemma}


\begin{proof}
Результаты леммы \ref{lemma4} следует из леммы \ref{lemma3} и того факта, что собственные значения матриц $D^{-1}Q(t)D$ и $Q(t)$ одни и те же. 
$\blacksquare$
\end{proof}


\begin{lemma}
\label{lemma5}
Пусть заданы $d_i>0$, $i=1,...,n$ и $\alpha \in [0; 1]$. 
Собственные значения матрицы $Q(t)$ из \eqref{eq_2_estim_1d0} находятся в области пересечения $e$-кругов 

\begin{equation}
\label{eq_2_estim_1}
\begin{array}{lll}
\mathcal{EC}^{D,\alpha}=\cup_{i=1}^{n} \mathcal{EC}^{D,\alpha}_{i},
\\
\end{array}
\end{equation}
где
\begin{equation*}
\begin{array}{lll}
\mathcal{EC}^{D,\alpha}_{i}=\cup_{\Delta q_{ii}(t) \in [\underline{q}_{ii};\overline{q}_{ii}]} \{ t \in \mathbb R_{\geq}, \lambda \in \mathbb C: |\lambda(t)-q_{ii}-\Delta q_{ii}(t)|
\leq
[\hat{R}^D_i(Q)]^{\alpha} [\hat{C}^D_i(Q)]^{1-\alpha}.
\end{array}
\end{equation*}

\end{lemma}


\begin{proof}
Доказательство леммы \ref{lemma5} следует из доказательств лемм \ref{lemma2} и \ref{lemma3} и того факта, что собственные значения матриц $D^{-1}Q(t)D$ и $Q(t)$ одинаковые. 
$\blacksquare$
\end{proof}


\begin{corollary}
\label{corollary_2}
Аналогично леммам \ref{lemma1} и \ref{lemma2} можно выписать оценки на максимальные и минимальные значения собственных значений матрицы $Q(t)$ с использованием результатов лемм \ref{lemma3}-\ref{lemma5}. 
То есть существуют числа $d_i>0$, $i=1,...,n$ и $\alpha \in [0;1]$ такие, что справедливы следующие оценки
\begin{equation}
\label{eq_2_estim_1_d_est}
\begin{array}{lll}
\max\limits_{i} \{\sup\limits_{t}\{\Re \{\lambda_i\{Q(t)\}\} \} \} \leq \sigma^D_{\max} \{Q(t)\} \leq \sigma_{\max} \{Q(t)\} ,
\\
\min\limits_{i} \{\sup\limits_{t}\{\Re \{\lambda_i\{Q(t)\} \}\} \} \geq \sigma^D_{\min} \{Q(t)\} \geq  \sigma_{\min} \{Q(t)\},
\\
\\
\max\limits_{i} \{\sup\limits_{t}\{\Re \{\lambda_i\{Q(t)\} \} \} \} \leq \sigma^{D,\alpha}_{\max} \{Q(t)\} \leq \sigma_{\max}^{\alpha} \{Q(t)\} ,
\\
\min\limits_{i} \{\sup\limits_{t}\{ \Re \{\lambda_i\{Q(t)\} \} \} \} \geq \sigma^{D,\alpha}_{\min} \{Q(t)\} \geq  \sigma_{\min}^{\alpha} \{Q(t)\},
\end{array}
\end{equation}
где
\begin{equation}
\label{eq_2_estim_2}
\begin{array}{lll}
\sigma_{\max}(Q(t)) = \min \left\{ \max\limits_{i}\{q^0_{ii} + \Delta \overline{q}_{ii} + \hat{R}_i(Q) \}, 
\max\limits_{j}\{q^0_{jj} + \Delta \overline{q}_{jj} + \hat{C}_j(Q) \} \right\},
\\
\sigma_{\min}(Q(t)) = \max\{ \min\limits_{i}\{q^0_{ii} - \Delta \overline{q}_{ii} - \hat{R}_i(Q) \}, \min\limits_{j}\{q^0_{jj} - \Delta \overline{q}_{jj} - \hat{C}_j(Q) \} \},
\\
\sigma^D_{\max}(Q(t)) = \min\{ \max\limits_{i}\{q^0_{ii} + \Delta \overline{q}_{ii} + \hat{R}^D_i(Q) \}, 
\max\limits_{j}\{q^0_{jj} + \Delta \overline{q}_{jj} + \hat{C}^D_j(Q) \} \},
\\
\sigma^D_{\min}(Q(t)) =\max \{ \min\limits_{i}\{q^0_{ii} - \Delta \overline{q}_{ii} - \hat{R}^D_i(Q) \}, 
\min\limits_{j}\{q^0_{jj} - \Delta \overline{q}_{jj} - \hat{C}^D_j(Q) \} \},
\\
\\
\sigma_{\max}^{\alpha}(Q(t)) = \max\limits_{i,\alpha}\{q^0_{ii} + \Delta \overline{q}_{ii} + [\hat{R}_i(Q)]^{\alpha} [\hat{C}_i(Q)]^{1-\alpha} \},
\\
\sigma_{\min}^{\alpha}(Q(t)) = \min\limits_{i,\alpha}\{q^0_{ii} - \Delta \overline{q}_{ii} - [\hat{R}_i(Q)]^{\alpha} [\hat{C}_i(Q)]^{1-\alpha} \},
\\
\sigma^{D,\alpha}_{\max}(Q(t)) = \max\limits_{i}\{q^0_{ii} + \Delta \overline{q}_{ii} + [\hat{R}_i^D(Q)]^{\alpha} [\hat{C}_i^D(Q)]^{1-\alpha} \},
\\
\sigma^{D,\alpha}_{\min}(Q(t)) = \min\limits_{i}\{q^0_{ii} - \Delta \overline{q}_{ii} - [\hat{R}_i^D(Q)]^{\alpha}  [\hat{C}_i^D(Q)]^{1-\alpha}\}.
\end{array}
\end{equation}
\end{corollary}


\textit{Пример 2.} 
Рассмотрим две параметрически неопределенные матрицы $Q$ с постоянными и переменными параметрами в виде 
\begin{equation*}
\begin{array}{lll}
Q=\begin{bmatrix}
-1 & 0\\
0 & -1.5
\end{bmatrix}
+
\underbrace{\begin{bmatrix}
r_{11} & 2r_{12}\\
3r_{21} & 4r_{22}
\end{bmatrix}}_{\Delta Q},
\\
Q(t)=\begin{bmatrix}
-1 & 0\\
0 & -1.5
\end{bmatrix}
+
\underbrace{\begin{bmatrix}
\sin (t) & 2 \cos (1.5t)\\
3 \textup{sign}(\sin(2t)) & 4 \textup{sign}(\cos(1.7t))
\end{bmatrix}}_{\Delta Q(t)},
\end{array}
\end{equation*}
где $r_{ij}$, $i,j=1,2$ -- псевдослучайные числа равномерно распределенные на интервале $(-1;1)$. 
Рассмотрим $200$ реализаций для каждого $r_{ij}$. 
Матрицы $\Delta Q$ и $\Delta Q(t)$ имеют одинаковые $m_{ij}$, поэтому и оценки области локализации будут одинаковыми.

На рис.~\ref{Ex2a} представлена область локализации собственных значений $Q$ и $Q(t)$ с использованием результатов лемм \ref{lemma3}-\ref{lemma5} (серая область), где кружками и треугольниками изображены собственные значения матрицы $Q$ с постоянными параметрами, а кривыми линиями -- собственные значения матрицы $Q$ с нестационарными параметрами. 
На трех рисунках из четырех пары $e$-кругов совпали за счет варьирования $d_i$ и $\alpha$, поэтому на трех рисунках указаны только два $e$-круга.

\begin{figure}[H]
\begin{minipage}[H]{0.44\linewidth}
\center{\includegraphics[width=1\linewidth]{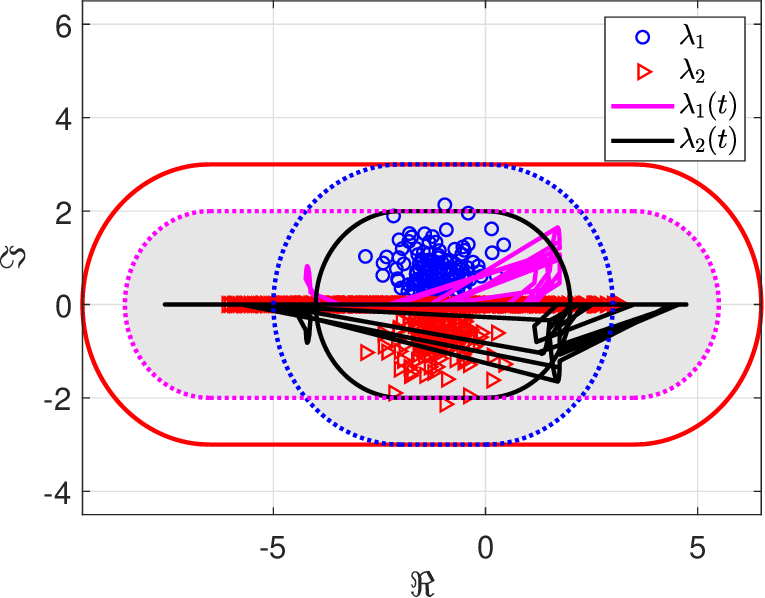}}
\end{minipage}
\hfill
\begin{minipage}[H]{0.44\linewidth}
\center{\includegraphics[width=1\linewidth]{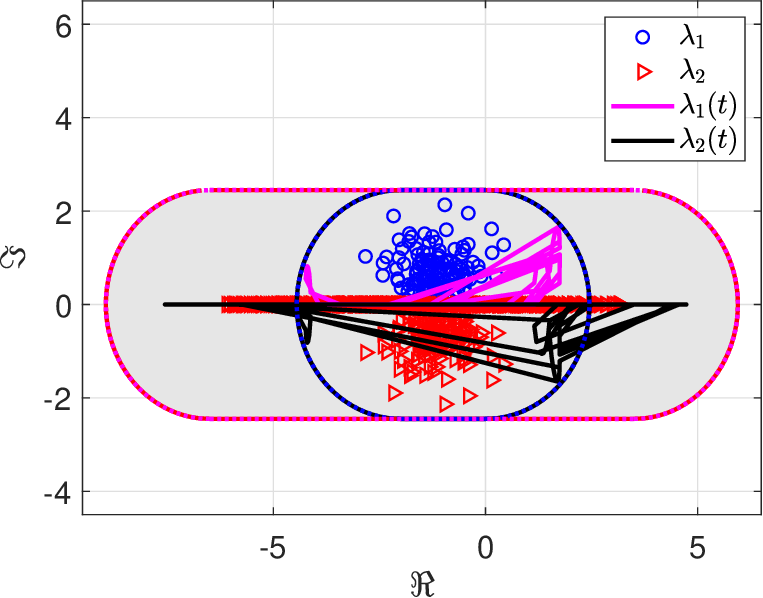}}
\end{minipage}
\vfill
\begin{minipage}[H]{0.44\linewidth}
\center{\includegraphics[width=1\linewidth]{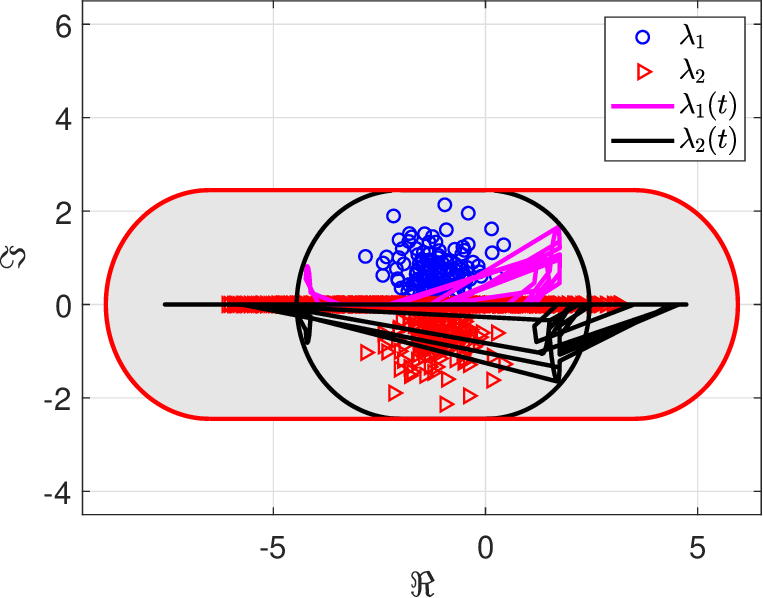}}
\end{minipage}
\hfill
\begin{minipage}[H]{0.44\linewidth}
\center{\includegraphics[width=1\linewidth]{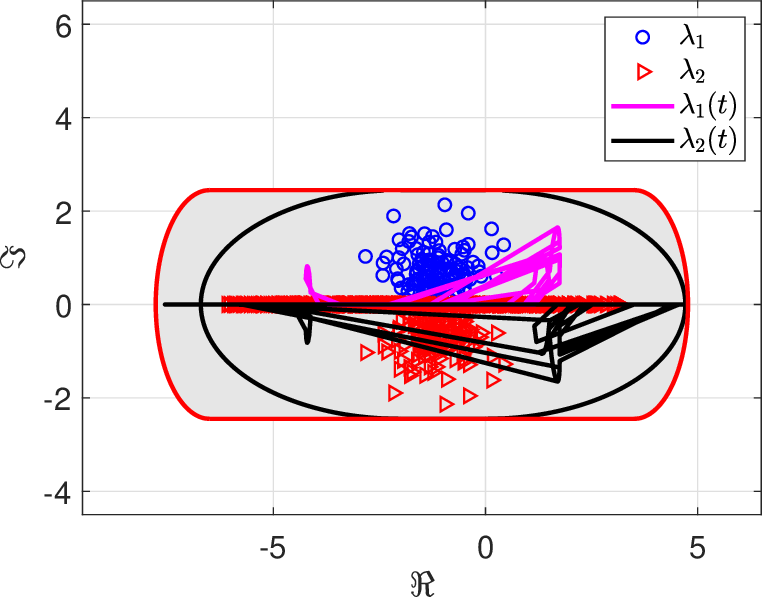}}
\end{minipage}
\caption{Области локализации собственных значений возмущенных матриц $Q$ и $Q(t)$.}
\label{Ex2a}
\end{figure}


%
%


\section{Анализ устойчивости систем управления}

В данном разделе рассмотрим несколько применений результатов предыдущего раздела для анализа и синтеза систем управления.

\subsection{Синхронизация сетевых систем}

Рассмотрим сетевую систему, состоящую из $n$ взаимосвязанных агентов вида
\begin{equation}
\label{eq_2_Ex_4_1}
\begin{array}{lll}
\dot{x}_i = \sum_{j=1}^{n} q_{ij} x_j + u_i,~~~i=1,...,n,
\end{array}
\end{equation}
где $x_i \in \mathbb R$, 
$u_i \in \mathbb R$ -- сигнал управления, 
$|q_{ij}| \leq m_{ij}$. 
Требуется обеспечить выполнение условия $\lim\limits_{t \to \infty} x_i(t)=0$ для всех $x_i$ за счет соответствующего выбора $u_i$, $i=1,...,n$. 

Зададим закон управления
\begin{equation}
\label{eq_2_Ex_4_u_i}
\begin{array}{lll}
u_i = - q x_i, ~~~ i=1,...,n,
\end{array}
\end{equation}
где $q>0$. 

Введем обозначения: $x=col\{x_1,...,x_n\}$, $Q_0=-qI$, $\Delta Q = (q_{ij})$ и $Q=Q_0+\Delta Q$. 
Тогда \eqref{eq_2_Ex_4_1} и \eqref{eq_2_Ex_4_u_i} можно переписать в виде
\begin{equation}
\label{eq_2_Ex_4_CLS}
\begin{array}{lll}
\dot{x} = Q x.
\end{array}
\end{equation}
В результате проверка условия $\lim\limits_{t \to \infty} x_i(t)=0$ сводится к проверке устойчивости матрицы $Q_0+\Delta Q$ которая может быть обеспечена за счет соответствующего выбора $q$ в \eqref{eq_2_Ex_4_u_i}.

Для иллюстрации сказанного, пусть $q=-10$ и $q_{ij}$ -- псевдослучайные числа равномерно распределенные на интервале $(-1;1)$. 
Для анализа устойчивости матрицы $Q_0+\Delta Q$ воспользуемся:
\begin{itemize}
\item функциями eig (расчет собственных значений матрицы) и lyap (решение уравнения Ляпунова) в MatLab, считая $q_{ij}$ известными;
\item приложениями CVX и Yalmip/Sedumi для расчета линейных матричных неравенств, считая $q_{ij}$ известными;
\item леммами \ref{lemma1} и \ref{lemma2}, считая $q_{ij}$ известными;
\item следствием \ref{corollary_2}, считая $q_{ij}$ неизвестными, но с известными $m_{ij}$.
\end{itemize}

На рис.~\ref{Ex3_fig} представлен график затрачиваемого времени на выполнение операции по определению устойчивости $Q_0+\Delta Q$ в зависимости от размерности данной матрицы (числа агентов в сети) и использовании соответствующего метода. 
Расчеты проводились в Matlab R2021b. 
Отметим, что при анализе предложенных результатов не было достигнуто предельное значение времени расчета из-за того, что Matlab R2021b не формировал матрицу размерности больше, чем $25000$.

\begin{figure}[h]
\center{\includegraphics[width=0.7\linewidth]{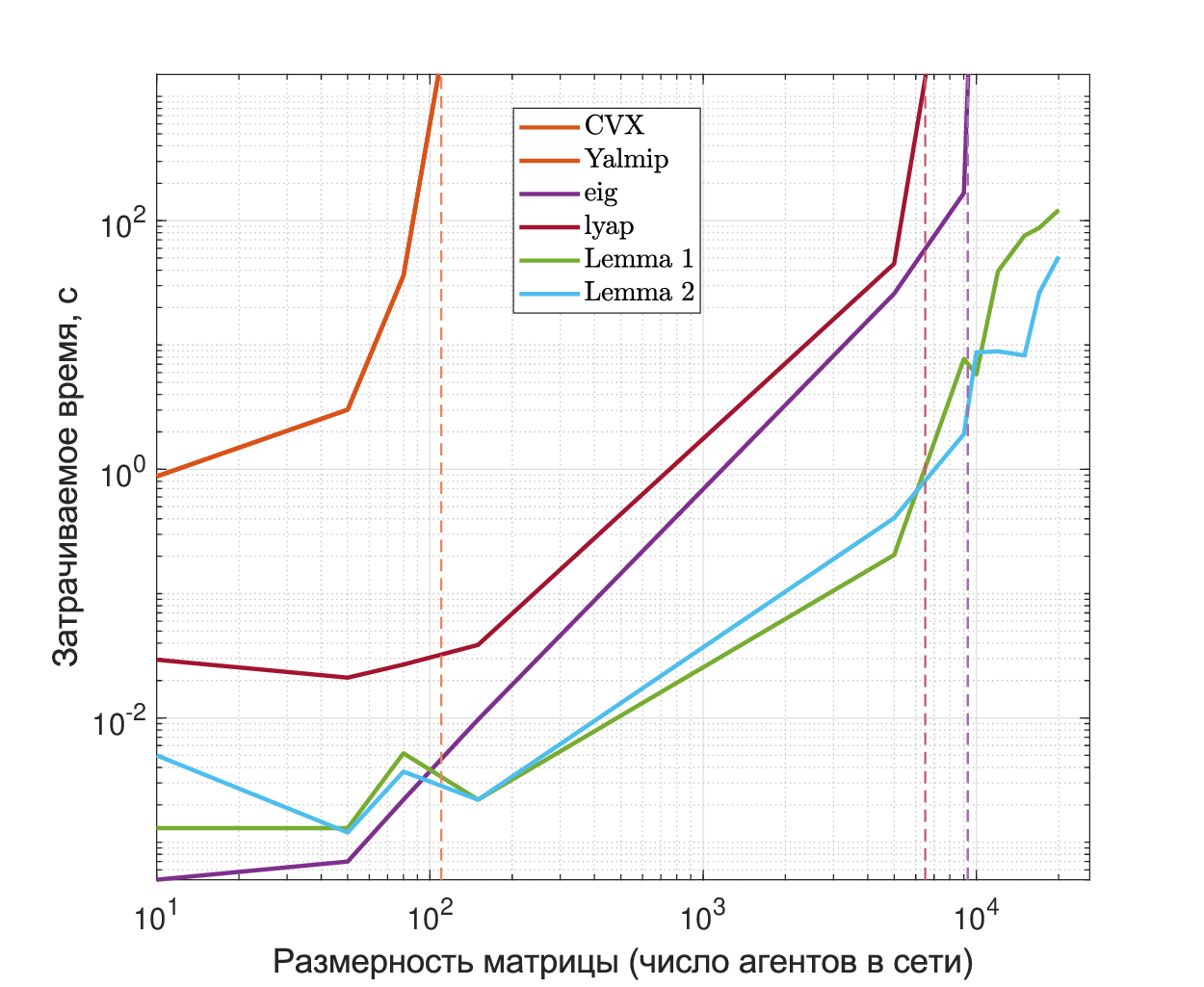}}
\caption{Зависимость затрачиваемого времени на выяснения устойчивости системы от числа агентов в сети.}
\label{Ex3_fig}
\end{figure}

%

Выводы:
\begin{itemize}
\item затрачиваемое время на расчет с использованием CVX и Yalmip существенно возрастает (в меньшей степени при использовании eig и lyap) при повышении количества агентов в сети, в то время как предложенные результаты наименее затратны по времени расчета;
\item алгоритмы eig, lyap, CVX и Yalmip/Sedumi обеспечивают высокую точность расчета по сравнению с предложенными результатами.
\end{itemize}

Замкнутая система \eqref{eq_2_Ex_4_CLS} содержит матрицу с диагональным преобладанием. 
В статьях \cite{Polyak02a,Polyak02b,Polyak04,Uronen72,Soloviev83,Curran09,Vijay15,Li19,Pachauri14,Xie22,Adom23,Kazakova98,Vijay14}, где также использовались матрицы с диагональным преобладанием, отмечалось, что это достаточно узкий класс исследуемых объектов. 

В следующих разделах покажем, что предложенные результаты можно применять к системам с матрицами без диагонального преобладания. 
Диагональное преобладание будет предъявляться к выражениям, полученным с использованием аппарата функций Ляпунова.


\subsection{Анализ устойчивости линейных систем с интервально неопределенными параметрами и матрицами без диагонального преобладания}

В данном разделе рассмотрим модификацию теоремы Демидовича \cite{Afanasiev03}(теорема 6.1), \cite{Khalil09} (в литературе также используется термин <<условие Демидовича>>) об исследовании устойчивости линейных систем с известными нестационарными параметрами на случай интервальной неопределенности и наличия внешних возмущений.
Пусть исследуемый объект представлен следующим уравнением
\begin{equation}
\label{eq_2_1}
\begin{array}{lll}
\dot{x}(t)=A(t)x(t)+F(t)f(t),
\end{array}
\end{equation}
где 
$t \geq 0$, 
$x \in \mathbb R^n$ -- вектор состояния, 
$f \in \mathbb R^l$ -- внешний сигнал такой, что $\sup\{|f(t)|\} \leq \bar{f}$, 
$F(t) \in \mathbb R^{n \times l}$ и $A(t)=(a_{ij}(t)) \in \mathbb R^{n \times n}$ такие, что 
$\sup\{\|F(t)\|\} \leq \bar{F}$, 
$A(t) = A_0 + \Delta A(t)$, 
$A_0 = (a^0_{ij})$, 
$\Delta A(t) = (\Delta a_{ij}(t))$, 
$\Delta \underline{a}_{ii} \leq \Delta a_{ii}(t) \leq \Delta \overline{a}_{ii}$, 
$|\Delta a_{ij}(t)| \leq m_{ij}$ при $i \neq j$.

Введем матрицу $\bar{A}(t)$, где
\begin{equation}
\label{eq_2_estim_1_upr_1}
\begin{array}{lll}
\bar{A}(t) = A(t)+A^{\rm T}(t) = \bar{A}_0+\Delta \bar{A}(t),
\\
\bar{A}_0 = (\bar{a}^0_{ij})=(a^0_{ij}+a^0_{ji}), 
\\
\Delta \bar{A}(t) = (\Delta \bar{a}_{ij}(t))=(\Delta {a}_{ij}(t)+\Delta {a}_{ji}(t)), 
\\
2\Delta \underline{a}_{ii} \leq \Delta \bar{a}_{ii}(t) \leq 2\Delta \overline{a}_{ii}, 
\\
|\Delta \bar{a}_{ij}(t)| \leq m_{ij}+m_{ji}~~~\mbox{при}~~~i \neq j.
\end{array}
\end{equation}


Отметим, что система \eqref{eq_2_1} содержит матрицу $A(t)$ без диагонального преобладания. 
Как будет показано в теореме ниже, диагональное преобладание понадобится в матрице $\bar{A}(t)$.

Согласно теореме Демидовича \cite{Afanasiev03}(теорема 6.1), \cite{Khalil09}, система \eqref{eq_2_1} асимптотически устойчива при $f(t) \equiv 0$ и с известной матрицей $A(t)$, если собственные значения матрицы $A(t)+A^{\rm T}(t)$ принимают отрицательные значения для всех $t$. 
Далее рассмотрим обобщение данной теоремы на интервально неопределенные матрицы с учетом следствия \ref{corollary_2}.

\begin{theorem}
\label{Th_1_1} 
Обозначим $\sigma$ любую из оценок сверху, вычисленную с помощью \eqref{eq_2_estim_1_d_est} для собственных значений матрицы $\bar{A}(t)$ в \eqref{eq_2_estim_1_upr_1}. 
Если $\sigma<0$, то справедлива следующая оценка
\begin{equation}
\label{eq_2_estim_3}
\begin{array}{lll}
|x(t)| \leq \frac{2 \|\bar{F}\| \bar{f}}{\sigma}+\mathcal{C} e^{0.5 \sigma t},
\end{array}
\end{equation}
где $\mathcal{C}=\max \left\{ 0, |x(0)|+\frac{2 \|\bar{F}\| \bar{f}}{\sigma} \right\}$.
\end{theorem}


\begin{proof}
Выберем функцию Ляпунова
\begin{equation}
\label{eq_2_V}
\begin{array}{lll}
V=x^{\rm T}x
\end{array}
\end{equation}
и найдем производную от нее вдоль решений \eqref{eq_2_1} в виде
\begin{equation}
\label{eq_2_V_dot}
\begin{array}{lll}
\dot{V}=x^{\rm T}\bar{A}(t)x + 2x^{\rm T}F(t)f.
\end{array}
\end{equation}

Найдем оценку сверху: 
\begin{equation}
\label{eq_2_V_dot}
\begin{array}{lll}
\dot{V} \leq
\sigma x^{\rm T}x + 2 |x| \|F(t)\| |f|
\leq 
\sigma V + 2 \sqrt{V} \|\bar{F}\| \bar{f}.
\end{array}
\end{equation}

Решим неравенство \eqref{eq_2_V_dot} в виде
\begin{equation}
\label{eq_2_V_solut1}
\begin{array}{lll}
\sqrt{V} \leq \frac{2 \|\bar{F}\| \bar{f}}{\sigma}+\left(\sqrt{V(0)}+\frac{2 \|\bar{F}\| \bar{f}}{\sigma} \right) e^{0.5 \sigma t}.
\end{array}
\end{equation}

Учитывая \eqref{eq_2_V}, получим
\begin{equation}
\label{eq_2_V_solut2}
\begin{array}{lll}
|x(t)| \leq \frac{2 \|\bar{F}\| \bar{f}}{\sigma}+\left(|x(0)|+\frac{2 \|\bar{F}\| \bar{f}}{\sigma} \right) e^{0.5 \sigma t}.
\end{array}
\end{equation}
Из выражения \eqref{eq_2_V_solut2} следуют результаты \eqref{eq_2_estim_3}.
$\blacksquare$
\end{proof}


\textit{Пример 3. Система с постоянными параметрами с матрицей без диагонального преобладания.} 
Рассмотрим систему \eqref{eq_2_1} с параметрами 
$A=\begin{bmatrix}
-1 & 3\\
-2.5 & -2
\end{bmatrix}$, 
$B=[0~0.05]^{\rm T}$ и $u=\sin(t)$. 
Матрица $A$ не является сверхустойчивой \cite{Polyak02a,Polyak02b,Polyak04} или с диагональным преобладанием \cite{Khorn89,Uronen72,Soloviev83,Curran09,Vijay15,Li19,Pachauri14,Xie22,Adom23,Kazakova98,Vijay14} не по строкам ни по столбцам. 
Также не существует $d_1>0$ и $d_2>0$, чтобы были выполнены условия \eqref{eq_2_estim_2}, так как неравенства $d_1-3d_2>0$ и $-2.5d_1+2d_2>0$, составленные для матрицы $A$, и неравенства $d_1-2.5d_2>0$ и $-3d_1+2d_2>0$, составленные для матрицы $A^{\rm T}$, не имеют решения. 

Рассмотрим матрицу $\bar{A}=A+A^{\rm T}=\begin{bmatrix}
-2 & 0.5\\
0.5 & -4
\end{bmatrix}$. 
Условие \eqref{eq_2_estim_2} будет выполнено для $\bar{A}$, где $\sigma=\sigma_{\max}(\bar{A})=-1.5$. 
Наибольшее собственное число матрицы $A+A^{\rm T}$ равно $-1.88$.
Если воспользоваться другим условием в \eqref{eq_2_estim_2} с $d_1=1$ и $d_2=0.711$, то оценку собственного числа можно улучшить до $\sigma=\sigma^D_{\max}(\bar{A})=-1.6445$.

\textit{Пример 4. Система с нестационарными параметрами с матрицей без диагонального преобладания.} 
Рассмотрим систему \eqref{eq_2_1} с параметрами с $A(t)=A_0 + \Delta A(t)$, где $A_0=A$ из предыдущего примера, 
$\Delta A(t) = 0.1
\begin{bmatrix}
\sin(t) & \cos(t)\\
\sin(2t) & \sin(4t)
\end{bmatrix}
$. 
Оценки сверху \eqref{eq_2_estim_1_upr_1} дают отрицательные значения, следовательно, система \eqref{eq_2_1} экспоненциально устойчива.


\subsection{Синтез закона управления для линейных систем с матрицами без диагонального преобладания}
\label{Sec3}

Рассмотрим систему
\begin{equation}
\label{eq_3_1_upr}
\begin{array}{lll}
\dot{x}(t)=A(t)x(t)+B(t)u(t)+F(t)f(t),
\end{array}
\end{equation}
где 
$u \in \mathbb R^m$ -- сигнал управления, 
$B(t) \in \mathbb R^{n \times m}$, $B(t)=b(t) B_0$, $\underline{b} \leq b(t) \leq \overline{b} \in \mathbb R$, $B_0$ -- известная матрица,
пара $(A(t),B(t))$ управляема для всех $t$. 
Остальные обозначения, как в \eqref{eq_2_1}. 
Предположим, что неизвестны параметры $\Delta A(t)$, $b(t)$, $F(t)$ и $f(t)$.

Введем закон управления
\begin{equation}
\label{eq_3_2_upr}
\begin{array}{lll}
u=Kx,
\end{array}
\end{equation}
где $K \in \mathbb R^{m \times n}$. 
Ниже сформулированы теоремы, позволяющие рассчитать матрицу $K$, которая обеспечивает экспоненциальную устойчивость замкнутой системы
\begin{equation}
\label{eq_3_2_CLS}
\begin{array}{lll}
\dot{x}(t)=(A(t)+B(t)K)x(t)+F(t)f(t).
\end{array}
\end{equation}

Отметим, что ни в матрице $A(t)$, не в матрице $A(t)+B(t)K$ не требуется выполнение свойства диагонального преобладания. 

\begin{theorem}
\label{Th_3_1aa}

Пусть матрицы $A$, $B$ и $F$ в \eqref{eq_3_1_upr} известны и постоянны, а также пусть 
для заданного $\alpha>0$ существует матрица $Q=Q^{\rm T}$ и коэффициент $\beta>0$ такие, что выполнены следующие условия
\begin{equation}
\label{eq_Psi_00}
\begin{array}{lll}
\Psi_{ii}<0, 
\\
\Psi_{ij} \geq 0 ~\mbox{при}~i \neq j, ~ i,j=1,...,n,
\\
\sigma(Q)>0,
\end{array}
\end{equation}
где 
\begin{equation}
\label{eq_Psi_1}
\begin{array}{lll}
\Psi=(\Psi_{ij})=QA^{\rm T}+AQ+Y^{\rm T}B^{\rm T}+BY+\alpha Q+\beta F^{\rm T}F,
\end{array}
\end{equation}
$\sigma(Q)$ -- одна из оценок снизу на собственные значения матрицы $Q$, полученная с помощью \eqref{eq_2_estim_2}. 
Тогда система \eqref{eq_3_1_upr} экспоненциально устойчива с $K=Y Q^{-1}$.
\end{theorem}


\begin{proof}
Выберем функцию Ляпунова
\begin{equation}
\label{eq_33_V}
\begin{array}{lll}
V=x^{\rm T}Px,
\end{array}
\end{equation}
где $P=Q^{-1}$, и найдем от нее производную по времени водоль решений \eqref{eq_3_2_CLS}:
\begin{equation}
\label{eq_33_V_dot}
\begin{array}{lll}
\dot{V}=x^{\rm T}[(A+BK)^{\rm T}P+P(A+BK)]x + 2x^{\rm T}Ff.
\end{array}
\end{equation}

Обозначив $z=col\{x,f\}$ и подставив \eqref{eq_33_V} и \eqref{eq_33_V_dot} в условие экспоненциальной устойчивости $\dot{V} + \alpha V + \gamma f^{\rm T}f < 0$, $\gamma>0$, получим
\begin{equation}
\label{eq_33_V_dot2}
\begin{array}{lll}
z^{\rm T}
\begin{bmatrix}
(A+BK)^{\rm T}P+P(A+BK)+\alpha P & PF\\
\star & -\gamma I
\end{bmatrix}
z<0.
\end{array}
\end{equation}

Следуя \cite{Boyd94}, неравенство \eqref{eq_33_V_dot2} будет выполнено, если будет выполнено следующее условие:
\begin{equation}
\label{eq_2_V_dot3}
\begin{array}{lll}
\begin{bmatrix}
(A+BK)^{\rm T}P+P(A+BK)+\alpha P & PF\\
\star & -\gamma I
\end{bmatrix}
<0.
\end{array}
\end{equation}

Воспользовавшись леммой Шура \cite{Boyd94}, перепишем \eqref{eq_2_V_dot3} в виде
\begin{equation}
\label{eq_2_V_dot4}
\begin{array}{lll}
(A+BK)^{\rm T}P+P(A+BK)+\alpha P + \beta PF^{\rm T}FP<0,
\end{array}
\end{equation}
где $\beta=1/\gamma$. 
Умножив слева и справа \eqref{eq_2_V_dot4} на $Q^{-1}$ и заменив $Y=KQ$, получим
\begin{equation}
\label{eq_2_V_dot5}
\begin{array}{lll}
\Psi:=QA^{\rm T}+AQ+Y^{\rm T}B^{\rm T}+BY+\alpha Q+\beta F^{\rm T}F<0.
\end{array}
\end{equation}

Согласно леммам \ref{lemma1} и \ref{lemma2}, собственные значения симметричных матриц $\Psi$ и $Q$ будут отрицательны и положительны соответственно, если будут выполнены неравенства \eqref{eq_Psi_00}. 
С другой стороны, согласно \cite{Khorn89} (теорема 7.2.1), эрмитова матрица положительно (отрицательно) определена в том и только в том случае, если все ее собственные значения положительны (отрицательны). 
Значит, условия $\Psi<0$ и $Q>0$ будут выполнены, если будут выполнены неравенства \eqref{eq_Psi_00}. 
$\blacksquare$
\end{proof}


С использованием результатов теоремы \ref{Th_3_1aa}, сформулируем следующую теорему для гораздо более широкого класса исследуемых систем с неизвестными нестационарным параметрами.


\begin{theorem}
\label{Th_3_1bb}
Рассмотрим систему \eqref{eq_3_1_upr} с нестационарными параметрами. 
Пусть существует матрица $Q=Q^{\rm T}$ и коэффициент $\beta>0$ такие, что выполнены условия
\begin{equation}
\label{eq_Psi_0}
\begin{array}{lll}
\Phi_{ii} < 0, 
\\
\Phi_{ij} \geq 0 ~\mbox{при}~i \neq j, 
\\
\sigma(Q) > 0,
\end{array}
\end{equation}
в вершинах $|\Delta a_{ij}(t)| \leq m_{ij}$ и $\underline{b} \leq b(t) \leq \overline{b}$, где 
\begin{equation}
\label{eq_Psi_1}
\begin{array}{lll}
\Phi=(\Phi_{ij})=&QA_0^{\rm T}+A_0Q+Q\Delta A^{\rm T}(t)+\Delta A(t) Q+
\\
&+b(t) Y^{\rm T} B_0^{\rm T}+b(t) B_0Y+\alpha Q+\beta \bar{F}^2 I,
\end{array}
\end{equation}
$\sigma(\Psi)$ -- одна из оценок сверху матрицы $\Psi$, полученная с помощью \eqref{eq_2_estim_1_d_est}. 
Тогда система \eqref{eq_3_1_upr} будет экспоненциально устойчивой с $K=Y Q^{-1}$, $P=Q^{-1}$.
\end{theorem}


\begin{proof} 
Воспользуемся результатами \eqref{eq_33_V}--\eqref{eq_2_V_dot5} из доказательства теоремы \ref{Th_3_1aa} с учетом нестационарных параметров. 
Поскольку $A(t)=A_0+\Delta A(t)$, $B(t)=b(t) B_0$ и $\|F(t)\| \leq \bar{F}$, то перепишем \eqref{eq_2_V_dot5} в виде
\begin{equation}
\label{eq_2_V_dot4_00}
\begin{array}{lll}
\Phi=&QA_0^{\rm T}+A_0Q+Q\Delta A^{\rm T}(t)+\Delta A(t) Q+
\\
&+b(t) Y^{\rm T} B_0^{\rm T}+b(t) B_0Y+\alpha Q+\beta \bar{F}^2 I<0.
\end{array}
\end{equation}
Если выполнены условия \eqref{eq_Psi_0} в вершинах $|\Delta a_{ij}(t)| \leq m_{ij}$ и $\underline{b} \leq b(t) \leq \overline{b}$ то, согласно \cite{Boyd94}, условие \eqref{eq_Psi_0} будет выполнено для любых $\Delta A(t)$ и $b(t)$ внутри политопа с вершинами $|\Delta a_{ij}(t)| \leq m_{ij}$ и $\underline{b} \leq b(t) \leq \overline{b}$.
$\blacksquare$
\end{proof}


\textit{Пример 5.} 
Рассмотрим систему \eqref{eq_3_1_upr} с параметрами
$A =
\begin{bmatrix}
0 & 1 & 0\\
0 & 0 & 1\\
1 & 2 & 3
\end{bmatrix}
$, 
$B = col\{0,0,1\}$,
$F = col\{0.1,0.5,1\}$ и $f(t)=\sin(t)$. 

Очевидно, что матрица $A$ без диагонального преобладания, а структура матрицы $B$ не позволяет законом управления $u=Kx$ с $K \in \mathbb R^{1 \times 3}$ привести к замкнутой системе с матрицей с диагональным преобладанием. 
Поэтому, воспользуемся теоремой \ref{Th_3_1aa} для анализа области локализации собственных значений матрицы $\Phi$, полученной в результате применения метода функций Ляпунова. 
Воспользовавшись теоремой \ref{Th_3_1aa}, получим $K =col\{-1.3671~-2.3619~-2.5724\}$ и $trace(P)=25.5858$. 
Воспользовавшись \cite{Topunov07}, получим $K =col\{-2.8862~-4.9244~-3.2136\}$ и $trace(P)=40.631$. 
В обоих случаях ставилась цель $trace(P) \to \min$ для расчета $K$. 

Из рис.~\ref{Ex5_fig} видно, что в установившемся режиме значение $|x(t)|$ предложенного алгоритма больше. 
Однако меньше всплеск $|x(t)|$ и амплитуда сигнала управления $u(t)$ в начальный момент времени, а также меньше значение $trace(P)$.

\begin{figure}[H]
\begin{minipage}[H]{0.44\linewidth}
\center{\includegraphics[width=1\linewidth]{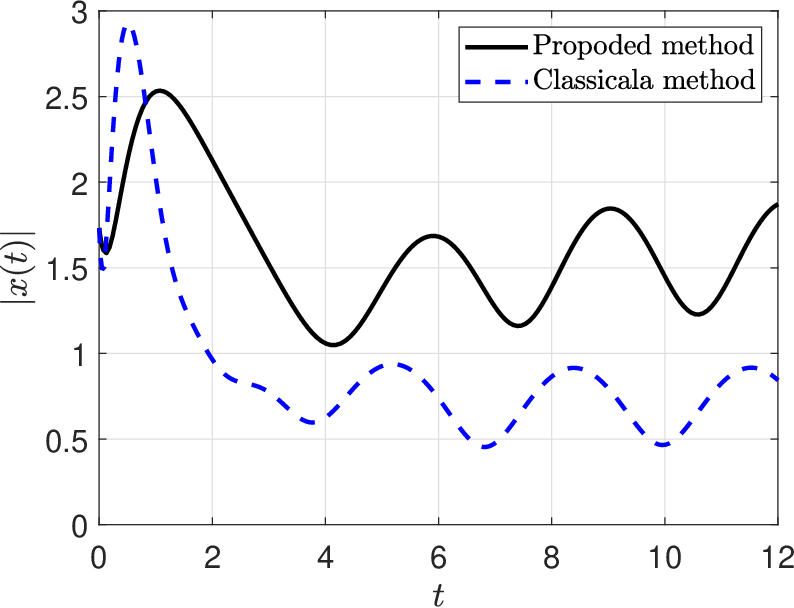}}
\end{minipage}
\hfill
\begin{minipage}[H]{0.44\linewidth}
\center{\includegraphics[width=1\linewidth]{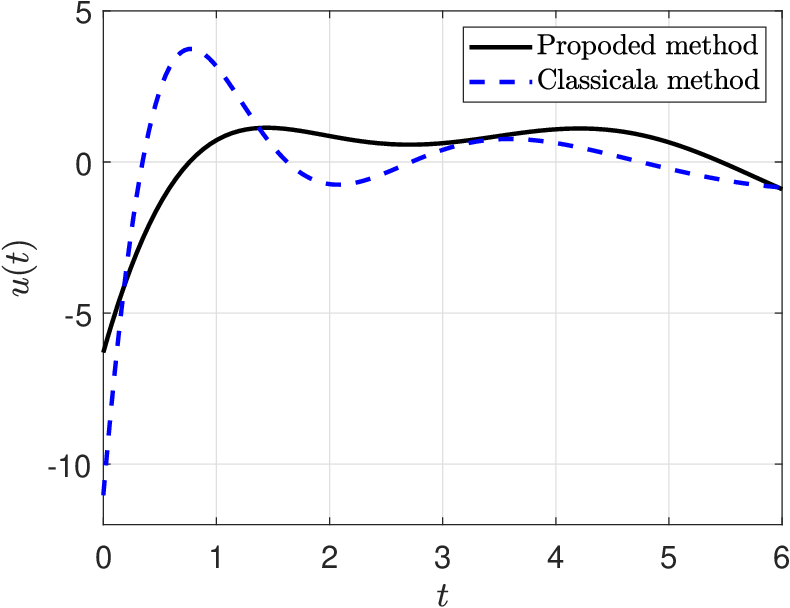}}
\end{minipage}
\caption{Переходные процессы по $|x(t)|$ и $u(t)$ для предложенного алгоритма (сплошные кривые) и алгоритма \cite{Topunov07} (пунктирные кривые).}
\label{Ex5_fig}
\end{figure}

\section{Заключение}
\label{Sec5}

В статье рассмотрено применение теоремы о кругах Гершгорина и производных от нее теорем для оценки области локализации собственных значений матрицы с постоянными и известными параметрами. 
Затем, данные результаты обобщаются на оценку области локализации для матриц с параметрической интервальной неопределенностью. 
Предложено понятие $e$-круга, позволяющее получить более точные оценки области локализации, чем прямое применение теоремы Гершгорина. 
Полученные результаты применены к управлению сетевыми системами, где показано, что для задач большой размерности предложенные результаты наименее затратны по времени выполнения операции по сравнению с процедурами eig и lyap (команды в MatLab для нахождения собственных значений матрицы и решения уравнения Ляпунова), а также cvx и yalmip, для решения линейных матричных неравенств. 
Предложено обобщение условия Демидовича для выяснения устойчивости нестационарной матрицы. 
Разработан подход по расчету матрицы в линейном законе управления при управлении линейными системами, где не выполнено свойство диагонального преобладания для матриц в замкнутой системе.

\end{document}